\documentclass{article}
\usepackage{spconf,amsmath,graphicx}
\usepackage{amsfonts, amsmath, amssymb, epsfig, graphicx, psfrag, subfigure, cite}
\usepackage{amsthm}
\newtheorem{proposition}{Proposition}

\usepackage{color}

\DeclareMathOperator*{\argmaxA}{\arg\max}
\title{Optimal Crowdsourced Classification with a Reject Option in the Presence of Spammers}
\name{Qunwei Li, Pramod K. Varshney}
\address{{ Department of EECS, Syracuse University, Syracuse, NY 13244 USA}\\
	{ \{qli33, varshney\}@syr.edu}}
\begin{document}
%
\maketitle
\begin{abstract}
We explore the design of an effective crowdsourcing system for an $M$-ary classification task. Crowd workers
complete simple binary microtasks whose results are aggregated to give the final decision. We consider
the scenario where the workers have a reject option so that they are allowed to skip microtasks when they are unable to or choose not to respond to binary microtasks. We present an aggregation approach using a weighted majority voting rule, where each worker's response is assigned an optimized weight to maximize crowd's classification performance.
\end{abstract}
\begin{keywords}
Classification, crowdsourcing, distributed inference, reject option, spammers
\end{keywords}
\section{Introduction}
\label{sec:intro}
{Crowdsourcing} provides a new framework to utilize distributed human wisdom to solve problems that machines cannot perform well, like handwriting recognition, paraphrase acquisition, audio transcription, and photo tagging \cite{paritosh2011computer,kamar2012combining,burrows2013paraphrase,7052378}. In spite of the successful applications of crowdsourcing, the relatively low quality of output remains a key challenge \cite{IpeirotisPW2010,allahbakhsh2013quality,mo2013cross}.

Several methods have been proposed to deal with the aforementioned problems \cite{KargerOS2011b,VempatyVV2014,yue2014weighted,6891807,VarshneyVV2014,QuinnB2011,zhang2012reputation,hirth2013analyzing}. A crowdsourcing task is decomposed into microtasks that are easy for an individual to accomplish, and these microtasks could be as simple as binary distinctions \cite{KargerOS2011b}. A classification problem with crowdsourcing, where taxonomy and dichotomous keys are used to design binary questions, is considered in \cite{VempatyVV2014}. New aggregation rules that mitigate the unreliability of the crowd and improve the crowdsourcing system performance are investigated in \cite{yue2014weighted,6891807}. In our research group, we employed binary questions and studied the use of error-control codes and decoding algorithms to design crowdsourcing systems for reliable classification \cite{VarshneyVV2014,VempatyVV2014}.  A group control mechanism where the reputation of the workers is taken into consideration to partition the crowd accordingly into groups is presented in\cite{QuinnB2011,zhang2012reputation}. Group control and majority voting techniques are compared in \cite{hirth2013analyzing}, which reports that majority voting is more cost-effective on less complex tasks. {A weighted voting framework is developed in \cite{Kuncheva2014}. However, prior information of the individual workers is assumed known, which is unrealistic in most practical situations.}

In past work on classification via crowdsourcing, crowd workers were required to provide a definitive yes/no response to binary microtasks. We consider the design of crowdsourcing systems where the workers are not forced to make a binary choice when they are unsure of their response and can choose not to respond. Crowd workers may be unable to answer questions for a variety of reasons such as lack of expertise. As an example, in mismatched speech transcription, i.e., transcription by crowd workers who do not know the language, workers may not be able to perceive the phonological dimensions they are tasked to differentiate \cite{jyothi2015acquiring}. We investigated the optimal aggregation rule where the workers have a reject option so that they are allowed to skip microtasks when they are unable to or choose not to respond \cite{7747496,Li:2017:CRC:3055601.3055607}.

In this paper, we extend our work \cite{7747496,Li:2017:CRC:3055601.3055607} by further taking the spammers' effect on the system into consideration. We study the scenario where spammers also exist in the crowd, who participate in the task only to earn some free money without regard to the quality of their answers. The spammers submit answers with random guesses. We propose an optimal aggregation rule to combat the spammers' effect on system performance, {which falls within the category of weighted majority voting methods, but where no prior individual information is needed.} Simulation results show significant performance improvement by the proposed method.

\section{Crowdsourcing with a Reject Option}
Consider the situation where $W$ workers take part in an $M$-ary object classification task. Each worker is asked $N$ simple binary questions, termed as microtasks, which eventually lead to a classification decision among the $M$ classes. We investigate independent microtask design and, therefore, we have $N = \left\lceil {{{\log }_2}M} \right\rceil $ independent microtasks of equal difficulty. The workers submit results that are combined to give the final decision. Here, we consider the microtasks as simple binary questions and the worker's answer to a single microtask is conventionally represented by either ``1'' (Yes) or ``0'' (No) \cite{VempatyVV2014,rocker2007paper}. Thus, the $w$th worker's ordered answers to all the microtasks form an $N$-bit word, which is denoted by ${\bf a}_w$. Let ${\bf a}_w(i)$, $i\in \{ 1,2,\dots,N\}$ represent the $i$th bit in this vector.  

In our previous work \cite{7747496,Li:2017:CRC:3055601.3055607}, we considered a more general problem setting where the worker has a reject option of skipping the microtasks. We denote this skipped answer as $\lambda$, whereas the ``1/0'' (Yes/No) answers are termed as definitive answers. Due to the variability of different workers' backgrounds, the probability of submitting definitive answers is different for different workers. Let $p_{w,i}$ represent the probability of the $w$th worker submitting $\lambda$ for the $i$th microtask. Similarly, let $\rho_{w,i}$ be the probability that ${\bf a}_w(i)$, the $i$th answer of the $w$th worker, is correct given that a definitive answer has been submitted. Due to the variabilities and anonymity of workers, we study crowdsourcing performance when $p_{w,i}$ and $\rho_{w,i}$ are realizations of certain probability distributions, which are denoted by ${F_P }\left( p \right)$ and ${F_\rho }\left( \rho \right)$ respectively. The corresponding means are expressed as $m$ and $\mu$.

Let $H_0$ and $H_1$ denote the hypotheses where ``0'' or ``1'' is the true answer for a single microtask, respectively. For simplicity of performance analysis, $H_0$ and $H_1$ are assumed equiprobable for every microtask. The crowdsourcing task manager or a fusion center (FC) collects the $N$-bit words from $W$ workers and performs fusion based on an aggregation rule. 

In our previous work \cite{7747496,Li:2017:CRC:3055601.3055607}, we proposed a novel weighted majority voting method for crowdsourced classification, which was derived by solving the following optimization problem
\begin{equation}\label{problem}
	\begin{array}{l}
		\text{maximize}\ \ {E_C}\left[ {{\mathbb{W}}} \right]\\
		\text{subject to}\ \ {E_O}\left[ {{\mathbb{W}}} \right] = {K}
	\end{array}
\end{equation}
where ${E_C}\left[ {{\mathbb{W}}} \right]$ denotes the crowd's average weight contribution to the correct class and ${E_O}\left[ {{\mathbb{W}}} \right]$ denotes the average weight contribution to all the possible classes that is constrained to remain a constant $K$. For $i$th bit, every worker's answer is assigned the derived optimal weight, and a decision is then obtained. We showed that this method significantly outperforms the widely-used simple majority voting procedure.

In this paper, we investigate the impact of spammers on system performance. The weight assignment scheme is developed by solving problem \eqref{problem} as well. 
\section{Optimal Behavior for the Manager}
In typical crowdsourcing setups, workers are simply paid in proportion to the number of tasks they complete \cite{shah2015double}. Most likely, the spammers will complete all the microtasks with random guesses. A payment mechanism was proposed in the crowdsourcing system with a reject option to incentivize the crowd, where responses with even the slightest error are associated with minimum payment possible \cite{shah2015double}. This mechanism promotes skipping of all the microtasks by the spammers. Therefore, we assume that $M_{A}$ spammers complete all the microtasks and the rest of the $M_0$ spammers skip all the microtasks, making a total of $M$ spammers in the crowd of size $W$.
To combat the spammers' effect on the system performance, we develop the aggregation rule on the manager's side with a new weight assignment scheme to maximize the weight assigned to the correct class. 
\begin{proposition}
	To maximize the average weight assigned to the correct classification element, the weight for the $w$th worker's answer is given by
	\begin{align}
		{W_w} = \left[\left( {W - M} \right){\mu ^n} +
		\frac{{{M_A}}}{{{2^N}{{\left( {1 - m} \right)}^N}}}\delta \left( {n - N} \right)\right]^{-1},
	\end{align}
	where $n$ is the number of definitive answers that the $w$th worker submits, and $\delta(\cdot)$ is the Dirac delta function.
\end{proposition}
\begin{proof}\renewcommand{\qedsymbol}{}
When there are $M$ spammers in the crowd with $M_0$ skipping and $M_A$ completing all the questions, the expected weight contributed to the correct class is given by
\begin{align}\label{21}
E_C[\mathbb{W}]=&\sum\limits_{w=1}^{W-M}E_{p,\rho}\left[\sum\limits_{n=0}^{N}W_w\rho(n)P_{\lambda}(n)\right]+\sum\limits_{w=1}^{M_0}W_w(n=0)\nonumber\\
&+\sum\limits_{w=1}^{M_A}\frac 1 {2^N}W_w(n=N)\nonumber\\
=&\sum\limits_{n=0}^N(W-M)W_w\mu^n\binom{N}{n}(1-m)^nm^{N-n}\nonumber\\
&+\sum\limits_{n=0}^NM_0W_w\delta(n)+\sum\limits_{n=0}^N\frac{M_A}{2^N}W_w\delta(n-N)\nonumber\\
=&\sum\limits_{n=0}^N(W-M)W_w\mu^n\mathbb P(n)+\sum\limits_{n=0}^N\frac{M_0}{\mathbb P(0)}W_w\mathbb P(n)\delta(n)\nonumber\\
&+\sum\limits_{n=0}^N\frac{M_A}{2^N \mathbb P(N)}W_w\mathbb P(n)\delta(n-N)\nonumber\\
=&\sum\limits_{n=0}^NW_wS(n)\mathbb P(n),
\end{align}
where $\mathbb P(n)=\binom{N}{n}(1-m)^nm^{N-n}$,
and 
\begin{align}
S(n)=(W-M)\mu^n+\frac{M_0}{m^N}\delta(n)+\frac{M_A}{2^N(1-m)^N}\delta(n-N).\nonumber
\end{align}

Note that ${\sum\limits_{n = 0}^N {\mathbb P(n)}  }=1$, and then \eqref{21} is upper-bounded using Cauchy-Schwarz inequality as follows:
\begin{align}
&E_C[\mathbb W]=\sum\limits_{n=0}^NW_wS(n)\mathbb P(n)\nonumber\\
\label{24}&\le \sqrt{\sum\limits_{n=0}^N(W_wS(n))^2\mathbb P(n)}\sqrt{\sum\limits_{n=0}^N\mathbb P(n)}=\alpha.
\end{align}
Also note that equality holds in \eqref{24} only if 
\begin{align}
W_wS(n)\sqrt{\mathbb P(n)}=\alpha\sqrt{\mathbb P(n)}\nonumber
\end{align}
where $\alpha$ is a positive constant, and $W_wS(n)=\alpha$

Therefore, the optimal behavior for the manager in terms of the weight assignment is obtained as
\begin{align}
{W_w} \!=\!\! \left[\!\left( {W \!-\! M} \right){\mu ^n}\! + \!\frac{{{M_0}}}{{{m^N}}}\delta \left( n \right)\! +\! \frac{{{M_A}}}{{{2^N}{{\left( {1 \!- \!m} \right)}^N}}}\delta \left( {n \!-\! N} \right)\!\right]^{-1}\!\!\!.\nonumber
\end{align}

Note that if a worker submits no definitive answers, i.e. $n=0$, the corresponding weight assigned is $(W-M+\frac{M_0}{m^N})^{-1}$. However, since this worker skips all the questions, his/her decision for a certain question is not taken into consideration at the fusion center and, without loss of generality, the corresponding weight can be set equal to zero. Therefore, the weight assignment for the scheme can be expressed as 
\begin{align}
{W_w} = \left[\left( {W - M} \right){\mu ^n} +
\frac{{{M_A}}}{{{2^N}{{\left( {1 - m} \right)}^N}}}\delta \left( {n - N} \right)\right]^{-1}.\nonumber
\end{align}
\end{proof}

{Compared to the weight assignment for an honest crowd \cite{7747496}, the derived scheme differs in terms of the weight assigned to the workers who complete all the microtasks. If the spammers skip all the microtasks, the weight assignment scheme remains the same, which is intuitively true as no random guesses are received by the manager from the spammers and the crowd can be considered as honest as well. Otherwise, the weight assignment scheme differs from the scheme given in \cite{7747496}.}

\subsection{Parameter Estimation}
In order to act optimally, the manager has to estimate several parameters before the weight assignment can be adopted. Specifically, one has to estimate $\mu, m,  M_A, M_0$ before he/she can proceed with the optimal weight assignment. We can employ either the \textit{Training-based} or \textit{Majority-voting based} method to estimate $\mu$ as stated in our previous work \cite{7747496}. Calculating the ratio of the sum of skipped questions over all the questions attempted by the crowd gives the estimated $m$. Based on the analysis in previous sections, the answers with all questions completed or skipped should be discarded for estimation.

We hereby jointly address the estimation of $M_0$ and $M_A$ by using the maximum likelihood estimation (MLE) method. First, as we employ $G$ gold standard questions, a worker has to respond to $N+G$ microtasks. Let $W_{N+G}$ denote the number of workers submitting $N+G$ definitive answers, and $W_0$ denote the number of workers skipping all the microtasks. Given the numbers of spammers respectively completing and skipping all the microtasks, $M_A$ and $M_0$, the joint probability distribution function of $W_{N+G}$ and $W_0$, $f(W_{N+G},W_0|M_A,M_0      )$,  is expressed in \eqref{mle}, where $\hat m$ is the estimated $m$. 

Therefore, by the MLE method, the estimates of $M_A$ and $M_0$, which are denoted by $\hat M_A$ and $\hat M_0$ respectively, can be obtained as
\begin{align}
	\left\{ {\hat  M_A,\hat M_0 } \right\}=\argmaxA_{\left\{{  M_A, M_0 }  \right\} \ge 0 }f(W_{N+G},W_0|M_A,M_0      ) .
\end{align}
\begin{figure*}[!ht]
	\normalsize
	\begin{align}\label{mle}
		f(W_{N+G},W_0|M_A,M_0      )=&\binom{W_0-M_0}{W-M_0-M_A}(\hat m^{N+G})^{W_0-M_0}(1-\hat m^{N+G})^{W-W_0-M_A}\nonumber\\
		&\cdot \binom{W_{N+G}-M_A}{W-W_0-M_A}(1-\hat m)^{(N+G)(W_{N+G}-M_A)}\left(1-(1-\hat m)^{N+G}\right)^{W-W_{N+G}-W_0}
	\end{align}
	\hrulefill \vspace*{4pt}
\end{figure*}

Once the manager has the estimation results $\hat \mu$, $\hat{m}$, ${\hat M_A}$, and ${\hat M_0}$, he/she can optimally assign the weight to the workers' answers for aggregation.

\subsection{Performance Analysis}
In this section, we characterize the performance of such a crowdsourcing classification framework, where the task manager behaves optimally, in terms of the probability of correct classification $P_c$. Note that we have an overall correct classification only when all the bits are classified correctly.

\begin{proposition}
	The probability of correct classification $P_c$ in the crowdsourcing system is 
	\begin{align}
		{P_c} =\Big[ \frac{1}{2} &+ \frac{1}{2}\sum\limits_{S} {\binom{W}{\mathbb{Q}}} \left( {F(\mathbb{Q}) - F^{\prime}(\mathbb{Q})} \right) \nonumber\\&+\frac{1}{4}\sum\limits_{S^\prime} {\binom{W}{\mathbb{Q}}} \left( {F(\mathbb{Q}) - F^{\prime}(\mathbb{Q})} \right)\Big] ^N
	\end{align}
	with
	\begin{align}
	F({{\mathbb Q}})\!=\! m^{q_0}\!\prod\limits_{n = 1}^N \!\!{{{\left( {1 \!-\! \mu } \right)}^{{q_{ - n}}}}{\mu ^{{q_n}}}{{\left( {C_{N - 1}^{n - 1}{{\left( {1 \!-\! m} \right)}^n}{m^{N - n}}} \right)}^{{q_{ - n}} + {q_n}}}} \nonumber
	\end{align}
	and
	\begin{align}
		F^{\prime}({\mathbb Q})\! =\! m^{q_0}\!\prod\limits_{n = 1}^N\! {{{\left( {1 \!-\! \mu } \right)}^{{q_n}}}{\mu ^{{q_{ - n}}}}{{\left( {C_{N - 1}^{n - 1}{{\left( {1 \!-\! m} \right)}^n}{m^{N - n}}} \right)}^{{q_{ - n}} + {q_n}}}} \nonumber
	\end{align}
	where
	\begin{align}
		{{\mathbb Q}} = \{ &({q_{ - N}},{q_{ - N + 1}}, \ldots {q_N},M_A^{\prime},M_A^{\prime \prime}):\nonumber\\
		&\sum\limits_{n =  - N}^N {{q_n} = W - M_A -M_0} ,M_A^{\prime}+M_A^{\prime \prime}=M_A \},\nonumber
	\end{align} with natural numbers $q_n$, $M_A^{\prime}$, and $M_A^{\prime \prime}$, 
	\begin{align}
		{S} \!=\! \left\{\! {{{\mathbb Q}}\!:\!\sum\limits_{n = 1}^N {\frac{{{q_n} - {q_{ - n}}}}{(W\!-\!M)\mu^n}} \! +\!(M_A^{\prime}\!-\!M_A^{\prime \prime})\frac{2^N(1\!-\!m)^N}{M_A} }\!>\!0 \right\},\nonumber
	\end{align}
	\begin{align}
		{S}^\prime  = \left\{\! {{{\mathbb Q}}\!:\!\sum\limits_{n = 1}^N {\frac{{{q_n} \!-\! {q_{ - n}}}}{(W\!-\!M)\mu^n}}  \!+\!(M_A^{\prime}\!-\!M_A^{\prime \prime})\frac{2^N(1\!-\!m)^N}{M_A} }\!=\!0 \right\},\nonumber
	\end{align} and $\binom{W}{\mathbb{Q}} = \frac{{W!}}{{\prod_{n =  - N}^N {{q_n}!} }}$.
\end{proposition}

\begin{proof}\renewcommand{\qedsymbol}{}
Due to the space limit, we only give the result here. The proof will be given in the extended version of the paper and a similar proof can be found in our previous paper \cite{7747496}.
\end{proof}
\subsection{Simulation Results}
In this section, we present the simulation results to illustrate the performance of the proposed schemes. $W=50$ workers participate in a crowdsourcing task with $N=3$ microtasks and $G=3$ gold standard questions. $F_P(p)$ is chosen as a uniform distribution $U(0,1)$, and let $F_\rho(\rho)$ be a uniform distribution expressed as $U(x,1)$ with $0\le x \le 1$, and thus we can have $\mu$ varying from 0.5 to 1.

\begin{figure}[h]
	\centering
	\includegraphics[width=2.8in]{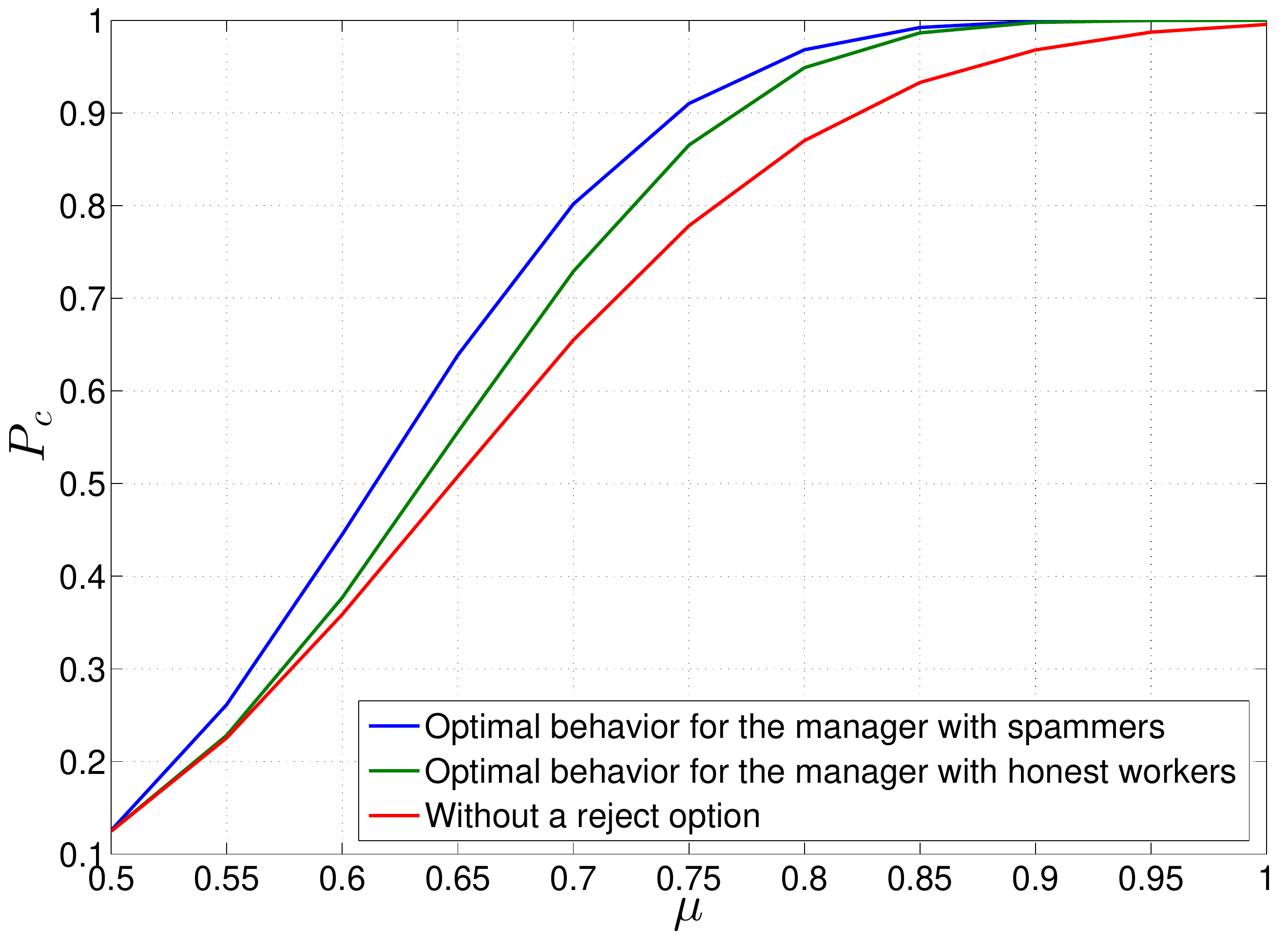} 
	\caption{Performance comparison with spammers.}
	\label{Pcwithvariousmu_estimationofM0andMN}
\end{figure}

We present the performance comparison with spammers in Fig. \ref{Pcwithvariousmu_estimationofM0andMN}, where the quality of the crowd $\mu$ varies. We plot the performance of three different weight assignment methods. The first one is what we derived in this section, which is referred to as the optimal behavior for the manager with spammers. The second is the one that we derived in \cite{7747496}, which is given by $W_w=\mu^{-n}$. {Since we do not assume the knowledge of prior information regarding individuals, the existing weighted majority voting methods fail to work in this setting. Thus, we choose the conventional simple majority voting without a reject option for comparison.} For illustration, there are 14 spammers in a crowd of 50 workers, and we have 7 spammers completing all the questions and the other 7 skipping all the questions. When $\mu=0.5$, the workers are making random guesses even if they believe that they are able to respond with definitive answers. In such a case, the choice of weight assignment schemes does not make a difference, and therefore, the three curves merge at this point. The method with optimal behavior for the manager with spammers outperforms the other two, while the simple majority voting without a reject option performs the worst.

\begin{figure}[h]
	\centering
	\includegraphics[width=2.8in]{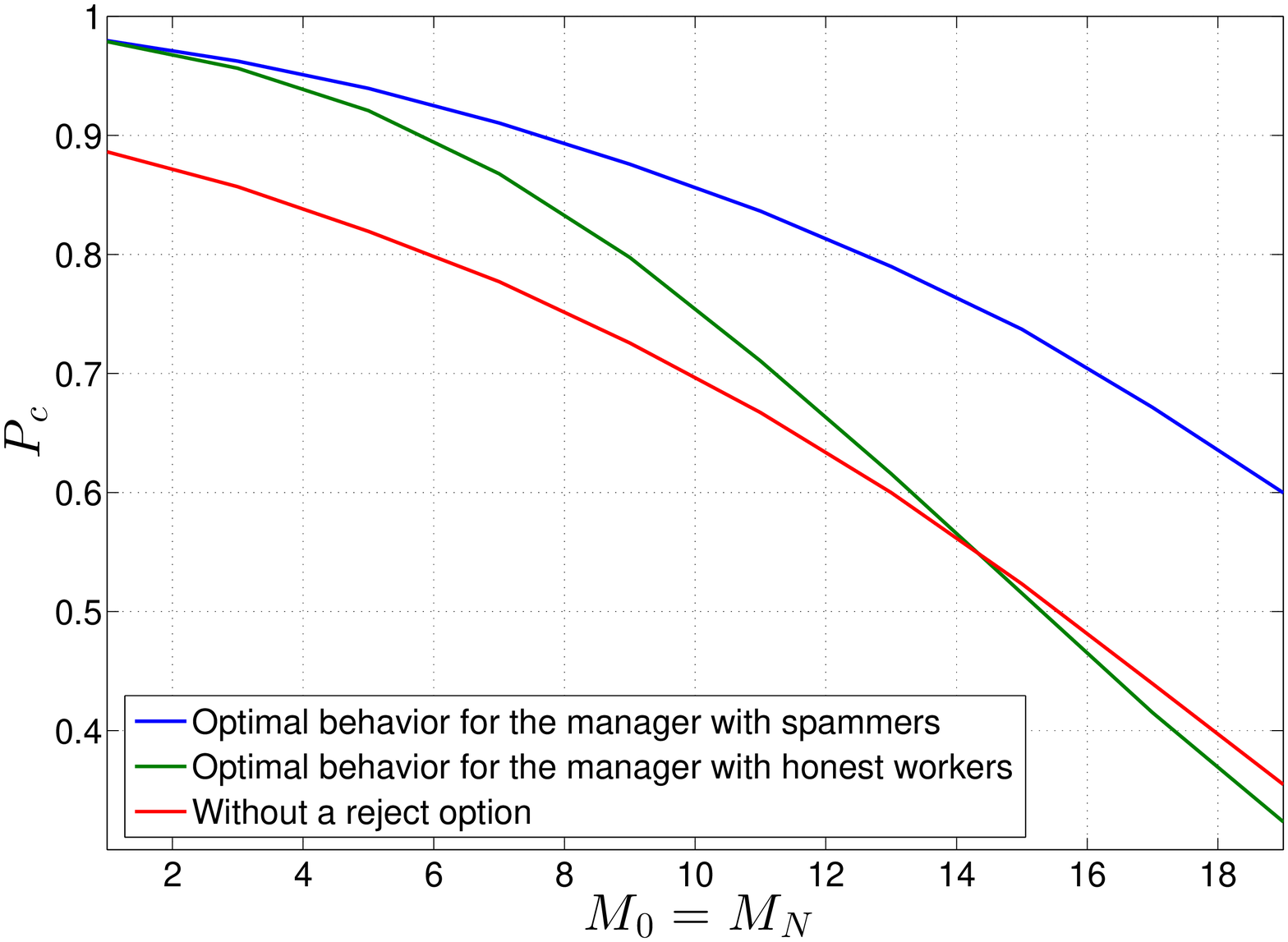} 
	\caption{Performance comparison with various spammers.}
	\label{Pcvariousspammers}
\end{figure}

In Fig. \ref{Pcvariousspammers}, we plot the performance comparison when the number of spammers changes. We set that $M_0=M_A$, and $\mu$ is fixed at 0.75. As we can observe, the method with optimal behavior for the manager with spammers yields the best performance. When the number of spammers is small, the simple majority voting method is outperformed by the one with optimal behavior for the manager with honest workers. However, this is not the case when the number of spammers is large. The reason is that with honest workers, the manager assigns a greater weight to the worker with a larger number of definitive answers. In the regime where $M_A$ is large, which means the number of spammers completing all the questions is large, the impact from the spammers is much more severe on the performance with such a weight assignment scheme. Thus, the corresponding performance degrades significantly.

\section{Conclusion}
We have studied a novel framework of crowdsourcing system for classification, where an individual worker has the reject option and can skip a microtask if he/she has no definitive answer. We investigated the impact of the spammers in the crowd on the crowdsourcing system performance. We derived the optimal strategy for the manager, where an optimal weighted aggregation rule for the crowdsourcing was proposed to combat the spammers' influence.

	\bibliographystyle{IEEEtran}
	\bibliography{IEEEabrv,ref_Lqw}

\end{document}